\documentclass[]{IEEEtran}
\usepackage{amsfonts}
\usepackage{amsmath}
\usepackage{multirow}
\usepackage{amsfonts}
\usepackage{amsmath,epsfig}
\usepackage{stfloats}
\usepackage{amssymb}
\usepackage{url}
\usepackage{bm}
\usepackage{cite}
\usepackage{amsmath}
\usepackage{amssymb}
\usepackage{latexsym}
\usepackage{multirow}
\usepackage{epsfig}
\usepackage{graphics}
\usepackage{mathrsfs}
\usepackage{subfigure}
\usepackage[all]{xy}
\usepackage{multicol,lipsum}
\usepackage{commath}

\usepackage{textcomp}
\usepackage{pifont}
\usepackage{bbding}
\usepackage{stmaryrd}
\usepackage{amssymb}
\usepackage{amsfonts}
\usepackage{epic}
\usepackage{algpseudocode}
\usepackage{stfloats}
\usepackage{epstopdf}
\usepackage{bm}
\usepackage{xcolor}
\usepackage{mathrsfs}
\usepackage{pifont}
\usepackage{bbding}
\usepackage{amsmath,epsfig}
\usepackage{stmaryrd}
\usepackage{amsmath}
\usepackage{amssymb}
\usepackage{amsthm}

\usepackage{amsfonts}
\usepackage{epic}
\usepackage{graphicx}
\usepackage{subfigure}
\usepackage{enumerate}
\usepackage{latexsym}
\usepackage{epstopdf}
\usepackage{multirow}
\usepackage{stfloats}
\usepackage{bm}
\usepackage{booktabs}
\usepackage{color}
\usepackage{cases}
\usepackage{array}
\usepackage{makecell}
\usepackage{times}
\usepackage[ruled,linesnumbered]{algorithm2e}
\usepackage{booktabs}
\usepackage{stackengine}
\usepackage{tikz}

\newtheorem{proposition}{Proposition}

\graphicspath{{./figs/}}

\begin{document}
\title{Pilot Decontamination for Massive MIMO Network with UAVs}
\author{\IEEEauthorblockN{Rui Lu, Qingqing Wu, \emph{Member, IEEE}, and Rui Zhang, \emph{Fellow, IEEE}
\thanks{R. Lu is with the Faculty of Electronic and Information Engineering, Xi'an Jiaotong University, Xi'an 710049, China (email:lorryxj@outlook.com). Q. Wu and R. Zhang are with the Department of Electrical and Computer Engineering, National University of Singapore, Singapore 117583 (email:\{elewuqq, elezhang\}@nus.edu.sg).}}}
\maketitle

\begin{abstract}
This letter studies the pilot contamination (PC) problem for massive multiple-input multiple-output (MIMO) networks with coexisting terrestrial users and unmanned aerial vehicles (UAVs). Due to the strong line-of-sight (LoS) air-to-ground channels between UAVs and base stations (BSs), UAVs usually cause a more severe PC issue as compared to the traditional terrestrial users.
To mitigate the PC caused by UAVs, we propose a low-complexity distributed scheme by exploiting the full-dimensional beamforming of massive MIMO BSs and the angle-dependent LoS channels between them and high-altitude UAVs.
Numerical results show the effectiveness of the proposed pilot decontamination scheme and the significant signal-to-interference-plus-noise ratio (SINR) gains in both the uplink and downlink after pilot decontamination.
\end{abstract}

\section{Introduction}\label{sct_introduction}
Massive multiple-input multiple-output (MIMO) is a promising solution to enable high-capacity communications for not merely traditional ground user equipments (GUEs), but also the new and emerging aerial users such as unmanned aerial vehicles (UAVs) \cite{zeng2019accessing,MassiveMIMO_UAV_EL_18,huang_beamtracking_wctl_2020,duo_bin_tvt_2020}. 
In future massive MIMO networks with coexisting UAVs and GUEs, UAVs may cause/suffer severe interference to/from a large number of non-associated base stations (BSs) due to the strong line-of-sight (LoS)-dominant UAV-BS channels.
Although massive MIMO processing at BSs can effectively mitigate the co-channel interference, its performance critically depends on the accuracy of the spatial channel state information (CSI) at the BSs.
Furthermore, the network throughput is fundamentally limited by the estimated CSI errors due to the pilot reuse over adjacent cells, the so-called \emph{pilot contamination (PC)} problem, even when the number of BS antennas goes to infinity \cite{PCP_MTL_TIT_18}.
In practice, the GUE-induced PC can be resolved if a sufficiently large pilot reuse factor is applied such that the same pilot can be avoided being reused by adjacent cells.
However, this method fails to deal with the UAV-induced PC due to the strong LoS-dominant UAV-BS channels, rendering that even two cells that are far apart may still suffer from the PC and its resultant interference.

Besides increasing the pilot reuse factor, other pilot decontamination schemes have also been proposed for terrestrial massive MIMO networks (see, e.g., \cite{PCP_MTL_TIT_18,SF_GongZJ_TWC_19,PC_coordinated_Yin_13} and references therein), while they face new challenges to mitigate the UAV-induced PC. 
For example, the large-scale fading precoding/decoding (LSFP and LSFD) algorithms in \cite{PCP_MTL_TIT_18} can eliminate PC by applying multi-cell cooperative processing. For terrestrial networks, the overhead of information exchange required for the cooperation is moderate because only a few BSs need to be coordinated. Whereas with UAVs, due to the LoS-dominant UAV-BS channels, much more BSs are required to participate in the cooperation, incurring prohibitive overhead in practical implementation. Similarly, the protocol-based scheme in \cite{SF_GongZJ_TWC_19} also faces this challenge, since the implementation of its required dynamic synchronization is more costly when a large number of BSs are involved. 
The coordinated pilot assignment scheme in \cite{PC_coordinated_Yin_13} can effectively eliminate PC by assigning the same pilot to non-spatially overlapped users, given that the covariance matrices of their channels are available.
Again, this scheme needs excessive BS cooperation for communicating with UAVs and it is also practically difficult to obtain the channel covariance matrices accurately.
Compared to the above schemes, assigning dedicated pilots to UAVs for their exclusive use may be a more practical solution to avoid PC between UAVs and GUEs, whereas PC still exists and needs to be resolved among UAVs. Besides, this approach will reduce the number of pilots available for GUEs, thus is not sustainable if the number of UAV users significantly grows in future wireless networks. 

In this letter, we first show analytically that the signal-to-interference-plus-noise ratio (SINR) performance will be significantly degraded for both the UAVs and GUEs due to the UAV-induced PC, even without considering the GUE-induced PC.
To resolve the UAV-induced PC, we further propose an efficient pilot decontamination (PDC) scheme by exploiting the angle-of-arrival (AoA)-dependent characteristics of UAV-BS channels, without the requirements of multi-cell cooperation and any prior channel statistical knowledge. Specifically, each BS first detects the LoS components from the least square (LS) channel estimates based on matched filtering. Then, the interfering ones are identified among the detected LoS components and further removed from the LS channel estimate.
The proposed scheme is practically appealing because BSs can perform pilot assignment and decontamination independently and UAVs are allowed to reuse pilots with GUEs. Simulation results validate the effectiveness of the proposed PDC scheme. 

\section{System Model and Pilot Contamination} \label{sct_sysmdl}
\subsection{System Model}
We consider a multi-cell massive MIMO network operating in time-division duplexing (TDD) mode to serve both GUEs and UAVs, as shown in Fig. \ref{fig_sysmdl}. 
Assume that the pilots used for the users in one cell are orthogonal, and each pilot group is reused by the users in some other cells. Moreover, UAVs are allowed to reuse pilots with GUEs.
Suppose that the pilot reuse factor is $ R $, e.g., $ R = 7 $ as shown in Fig. \ref{fig_sysmdl}, while the frequency reuse factor is $ 1 $. 
Without loss of generality, we focus on one particular pilot and denote the set of users sharing it by $ \mathcal{K} \triangleq \{1,2,\cdots,K\} $, where $ K $ is the number of the users using this pilot.
In addition, these $ K $ users are respectively associated with $ K $ BSs, with $ \mathcal{L} \triangleq \{1,2,\cdots,K\} $ denoting their set.
Suppose that among the $ K $ users, $ K_{\mathrm{u}} $ with $ 1 \leq K_{\mathrm{u}} < K $ users are UAVs, and we define $ \mathcal{U} \subseteq \mathcal{K} $ as the subset consisting of the $ K_{\mathrm{u}} $ UAVs.
Each BS is equipped with a uniform circular array (UCA) consisting of $ M $ antennas, while each user employs a single antenna for simplicity.
\setlength{\textfloatsep}{1pt}
\begin{figure}[t]
	\centering
	\includegraphics[width=0.5\textwidth,draft=false]{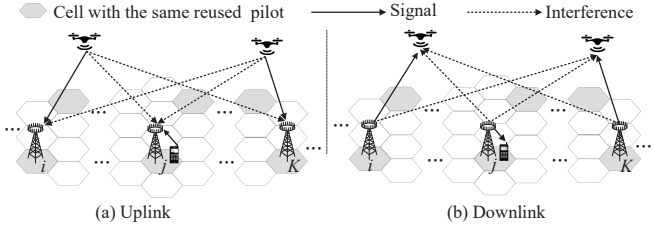}
	\caption{Pilot contamination in the cellular network with both UAVs and GUEs. \label{fig_sysmdl}}
\end{figure}

Denote the channel from the user $ k \in \mathcal{K} $ to the BS $ l \in \mathcal{L} $ by $ \mathbf{h}_{lk} \in \mathbb{C}^{M\times1} $, which is modeled by $ \mathbf{h}_{lk} = \sqrt{\beta_{lk}} \mathbf{g}_{lk} $,
where $ \beta_{lk} $ accounts for the large-scale path loss and $ \mathbf{g}_{lk} $ is a complex vector accounting for the small-scale fading. In this letter, we model the UAV-BS channels as LoS, and model the GUE-BS channels as Rayleigh fading. 
This is because we are mainly interested in the scenario that UAVs fly higher than BSs, in which their channels are dominated by LoS links in practice \cite{channel_model}. 
Thanks to the LoS propagation, for UAV-BS channels, $ {\mathbf{g}}_{lk} $ can be simplified as
$ {\mathbf{g}}_{lk} = \alpha_{lk} \mathbf{a}(\theta_{lk},\phi_{lk}) $, where $ \alpha_{lk} $ is a complex number denoting a random phase rotation with $ |\alpha_{lk}|^{2} = 1 $, and $ \mathbf{a}(\theta_{lk},\phi_{lk}) \in \mathbb{C}^{M \times 1} $ is the steering vector of the UCA at the BS with the $ m $-th element given by \cite[Section 4.2]{optimal_array_processing} 
\begin{equation}\label{equ_steering_vector}
\scalebox{0.99}{$[\mathbf{a}(\theta_{lk},\phi_{lk})]_{m} = \mathrm{exp} \left \{-j \frac{2 \pi d}{\lambda} \sin (\theta_{lk}) \cos (\phi_{lk} - \gamma_{m}) \right \} $}.
\end{equation}
%
In \eqref{equ_steering_vector}, $ \theta_{lk}$ and $ \phi_{lk} $ denote the elevation and azimuth angles of UAV $ k $ from BS $ l $, respectively, $ d $ denotes the radius of the BS UCA, $ \gamma_{m} = {2\pi(m-1)}/{M} $ denotes the angular location of antenna $ m $ on the UCA, and $ \lambda $ is the carrier wavelength. On the other hand, the small-scale fading components of GUE-BS channels are modeled as independent and identically distributed (i.i.d.) circularly symmetric complex Gaussian (CSCG) random variables with $ {\mathbf{g}}_{lk} \sim \mathcal{CN}(\mathbf{0}, \mathbf{I}) $, i.e., Rayleigh fading. 
Furthermore, based on the measurement results of air-to-ground channel (see, e.g., \cite{channel_model} and references therein), GUEs usually experience more severe path loss and shadowing than UAVs. 
As such, to focus on investigating the UAV-induced PC, we assume for simplicity that the GUE-induced PC (as well as its resultant interference) is negligible in this letter. This is practically valid since the GUE-induced PC can be solved by applying either the existing cooperative PDC schemes (see, e.g., \cite{PCP_MTL_TIT_18} and references therein) or a sufficiently large pilot reuse factor in the network.
%
\subsection{Uplink Channel Estimation and PC}\label{set_ULCE_PC}
As shown in Fig. \ref{fig_sysmdl}(a), to facilitate channel estimation, each user transmits the given pilot to its associated BS in the uplink training. Denote by $ \bm{\psi}_{0} \in \mathbb{C}^{\tau \times 1} $ the pilot sequence, with $ \bm{\psi}_{0}^{H} \bm{\psi}_{0} = 1 $ and $ \tau $ being the length of the pilot sequence. $ (\cdot)^{H} $ denotes the Hermitian transpose.
Then the received pilot signals during the pilot transmission, $ \mathbf{Y}_{l} \in \mathbb{C}^{M \times \tau} $, at BS $ l $ can be compactly expressed as
\begin{equation} \label{}
\mathbf{Y}_{l} = \sqrt{\tau p_{\mathrm{p}}}\mathbf{h}_{ll} \bm{\psi}^{T}_{0} + \sqrt{\tau p_{\mathrm{p}}} \sum_{k \in \mathcal{U}_{l}} \mathbf{h}_{lk} \bm{\psi}^{T}_{0} + \mathbf{N}_{l},
\end{equation}
where $ \mathcal{U}_{l} $ denotes the set of interfering UAVs to the user (a UAV or GUE) served by BS $ l $. To differentiate the two cases whether the user served by BS $ l $ is a UAV, we define $ \mathcal{U}_{l} $ as $ \mathcal{U}\setminus \{ l \} $ if $ l \in \mathcal{U} $, and $ \mathcal{U} $ otherwise. $ (\cdot)^{T} $ denotes the transpose.
In addition, $ {p_{\mathrm{p}}} $ denotes the pilot signal power, and $ \mathbf{N}_{l} \in \mathbb{C}^{M \times \tau} $ denotes the receiver noises, in which the elements are assumed to be i.i.d. CSCG random variables with zero mean and (normalized) unit variance.
By correlating the received signals $ \mathbf{Y}_{l} $ with $ \bm{\psi}_{0}^{*} $, we obtain the LS-based channel estimation given by
%
\begin{equation}\label{equ_chan_los_nlos}
		\hat{\mathbf{h}}_{ll} =\frac{\mathbf{Y}_{l} \bm{\psi}^{*}_{0}}{\sqrt{ \tau p_{\mathrm{p}}}} =\mathbf{{h}}_{ll} + \sum_{k \in \mathcal{U}_{l} } \alpha_{lk} \sqrt{\beta_{lk}} \mathbf{a}(\theta_{lk},\phi_{lk}) + \mathbf{n}_{l},
\end{equation}
where $ \mathbf{n}_{l} = {\mathbf{N}_{l} \bm{\psi}^{*}_{0}}/{\sqrt{ \tau p_{\mathrm{p}}}} \sim \mathcal{CN}(\mathbf{0}, \mathbf{I}/({\tau p_{\mathrm{p}}})) $. Note that the second term in (\ref{equ_chan_los_nlos}) is due to the strong LoS interference from the (other) UAVs. Thus, the PC problem arises due to the strong UAV interference, even under the assumption that the GUE-induced PC is already resolved (by e.g. using a sufficiently large pilot reuse factor).  

Now consider the uplink data transmission with the estimated channel in \eqref{equ_chan_los_nlos}.
Let $ {x}_{k} $ be the data sent by user $ k $ with $ \mathbb{E}[|x_{k}|^{2}] = 1, k \in \mathcal{K} $, where $ \mathbb{E}[\cdot] $ denotes the statistical expectation. The received signal at BS $ l $ can be expressed as\footnote{For simplicity, we consider the interference among the users using the same pilot only, while ignoring that from the other users using different pilots. This is because the interference from the users using different pilots vanishes as $ M \to \infty $ \cite{eng_spt_efficiency_2013_TCOM}. For finite $ M $, this assumption results in an upper bound for the users of interest since less interference is considered.}
\begin{equation}\label{key}
\mathbf{y}^{\mathrm{u}}_{l} = \sqrt{p_{\mathrm{u}}}\mathbf{h}_{ll} {x}_{l} + \sqrt{p_{\mathrm{u}}} \sum_{k \in \mathcal{U}_{l}} \mathbf{h}_{lk} x_{k} + \mathbf{n}^{\mathrm{u}}_{l},
\end{equation}
where $ p_{\mathrm{u}} $ is the transmit power for uplink data transmission, and $ \mathbf{n}^{\mathrm{u}}_{l} $ denotes the receiver noises with $ \mathbf{n}^{\mathrm{u}}_{l} \sim \mathcal{CN}(\mathbf{0},\mathbf{I}) $.
By applying maximum-ratio combining (MRC), the desired signal from the user served by BS $ l $ is given by $ \hat{x}_{l} = \overline{\mathbf{w}}_{l}^{H} \mathbf{y}^{\mathrm{u}}_{l} $, where $ \overline{\mathbf{w}}_{l} = {\hat{\mathbf{h}}_{ll}}/({\eta_{l} \sqrt{M}}) $ is the normalized combining vector with $ \eta_{l} = {||\hat{\mathbf{h}}_{ll}||}/{\sqrt{M}} $. Let $ p_{\mathrm{u}} = E_{\mathrm{u}}/{M} $ with $ E_{\mathrm{u}} $ being a constant regardless of $ M $ (for energy conservation with asymptotically large $ M $ \cite{eng_spt_efficiency_2013_TCOM}). Then, the uplink receive SINR for user $ l $ can be expressed as
\begin{equation}\label{equ_SINR_UL_GUE}
	\mathrm{SINR}^{\mathrm{u}}_{l} = \frac{{E_{\mathrm{u}}}/{\eta_{l}^{2}} \abs{{\hat{\mathbf{h}}_{ll}^{H} \mathbf{h}_{ll}}/{M}}^{2}}{\sum_{k \in \mathcal{U}_{l}} {E_{\mathrm{u}}}/{\eta_{l}^{2}} \abs{{\hat{\mathbf{h}}_{ll}^{H} \mathbf{h}_{lk}}/{M}}^{2} + 1}.
\end{equation}
To obtain its asymptotic value as $ M \to  \infty $, we introduce the following proposition. 
\begin{proposition}
	Let $ {\mathbf{p}}_{{i}} \in \mathbb{C}^{M \times 1} $ be a random vector whose elements are i.i.d. zero-mean random variables each with variance $ \sigma^{2} $, and for $ \bar{i} \neq i $, $ {\mathbf{p}}_{\bar{i}} $ and $ {\mathbf{p}}_{{i}} $ are independent. In addition, let $ {\mathbf{q}_{j}} \in \mathbb{C}^{M \times 1} $ be a steering vector given by \eqref{equ_steering_vector}, and for $ \bar{j} \neq j $, $ {\mathbf{q}}_{\bar{j}} $ and $ {\mathbf{q}}_{{j}} $ are associated with different AoAs. Then as $  M \to \infty $, we have
	\begin{equation}\label{equ_asm_rlt}
	\frac{{\mathbf{p}}^{H}_{i}{\mathbf{p}}_{\bar{i}}}{M} \xrightarrow[]{a.s} \delta_{i,\bar{i}} \sigma^{2}, \frac{{\mathbf{q}}^{H}_{j}{\mathbf{q}}_{\bar{j}}}{M} \xrightarrow[]{a.s} \delta_{j,\bar{j}},~{and}~
	\frac{{\mathbf{p}}^{H}_{i}{\mathbf{q}}_{{j}}}{M} \xrightarrow[]{a.s} 0,
	\end{equation}
	where $ \xrightarrow[]{a.s} $ denotes the almost sure convergence and $ \delta_{ij} $ is the Kronecker delta function.
\end{proposition}
\begin{proof}
	Note that the first two results in \eqref{equ_asm_rlt} follow from \cite[Section II-B]{eng_spt_efficiency_2013_TCOM} and \cite[Lemma 1]{Dai_Orth_LoS_ICC_2017}, respectively, and the third one can be obtained by the law of large numbers.
\end{proof}

Since the channel vectors are independent, by applying \eqref{equ_asm_rlt}, we have 
\begin{equation}\label{equ_h_asymptotic}
	\frac{{\mathbf{h}}^{H}_{lk}{\mathbf{h}}_{\bar{l}\bar{k}}}{M} \xrightarrow[]{a.s} \left\{ \begin{array}{ll}
	\beta_{lk}, & \mathrm{if}~\bar{l}={l}, \bar{k} = {k};\\
	0,  & \mathrm{otherwise}.
	\end{array} \right.
\end{equation}
Define $ \eta_{l,\infty}^{2} $ as the asymptotic value of $ \eta_{l}^{2} $ as $ M \to \infty $. Based on \eqref{equ_chan_los_nlos} and \eqref{equ_h_asymptotic}, we have $ \eta_{l,\infty}^{2} = \lim\limits_{M \to \infty} {\hat{\mathbf{h}}_{ll}^{H} \hat{\mathbf{h}}_{ll}}/{{M}} \xrightarrow[]{a.s} \beta_{ll} + \sum_{k \in \mathcal{U}_{l} }\beta_{lk} + {1}/({ \tau p_{\mathrm{p}}}), \forall l.$
Substituting \eqref{equ_chan_los_nlos} into \eqref{equ_SINR_UL_GUE} and using \eqref{equ_h_asymptotic}, we can obtain the asymptotic uplink SINR of the user served by BS $ l $, which is given by
\begin{equation}\label{equ_SINR_UL_GUE_M}
\mathrm{SINR}^{\mathrm{u}}_{l} \xrightarrow[M \to \infty]{} \frac{E_{\mathrm{u}} \beta^{2}_{ll}/\eta_{l,\infty}^{2}}{\sum_{k \in \mathcal{U}_{l}} E_{\mathrm{u}} \beta_{lk}^{2} /\eta_{l,\infty}^{2} + 1}.
\end{equation}
It is observed from \eqref{equ_SINR_UL_GUE_M} that as $ M \to \infty $, even using large pilot reuse factor to remove the GUE-induced PC, the user SINR is still limited by the interference from UAVs (if $ K_{\mathrm{u}} > 0 $). Besides, a GUE suffers the interference from all $ K_{\mathrm{u}} $ UAVs while a UAV suffers that only from $ K_{\mathrm{u}} - 1 $ UAVs.
%
\subsection{Downlink Data Transmission with Contaminated Channel}\label{sec_sys_dl}
For the downlink data transmission, each BS treats the contaminated channel estimate as the true channel and uses conjugate precoding to transmit signal to its associated user.
Denote by $ x_{l} $ the information-bearing symbol intended for the user served by BS $ l $, which satisfies $ \mathbb{E}[|x_{l}|^{2}] = 1, l \in \mathcal{L} $. First, consider a UAV in the downlink for which the received signal is expressed as
\begin{equation}
x_{i}^{\mathrm{d}} = \sqrt{p_{\mathrm{d}}} \mathbf{h}_{ii}^{T} \underline{{\mathbf{w}}}_{i} x_{i} + \sqrt{p_{\mathrm{d}}} \sum_{l \in \mathcal{L} \setminus \{i\}} \mathbf{h}_{li}^{T} \underline{{\mathbf{w}}}_{l} x_{l} + {n}^{\mathrm{d}}_{i},
 \end{equation}
where $ p_{\mathrm{d}} $ is the downlink transmit power and $ \underline{{\mathbf{w}}}_{l} = {{\hat{\mathbf{h}}}_{ll}^{*}}/({\eta_{l} \sqrt{M}}) $ is the precoding vector with $ \eta_{l} = {||\hat{\mathbf{h}}_{ll}||}/{\sqrt{M}} $. $ {n}^{\mathrm{d}}_{i} $ is the receiver noise with $ {n}^{\mathrm{d}}_{i} \sim \mathcal{CN}(0,1) $.
Similar to Section \ref{set_ULCE_PC}, let $ p_{\mathrm{d}} = E_{\mathrm{d}}/{M} $ with $ E_{\mathrm{d}} $ being a constant regardless of $ M $. Then the downlink receive SINR can be expressed as 
\begin{equation}\label{equ_SINR_DL_UAV}
\begin{split}
	\mathrm{SINR}^{\mathrm{d}}_{i} = \frac{{E_{\mathrm{d}}}/{\eta_{i}^{2}} \abs{{\mathbf{h}_{ii}^{H} \hat{\mathbf{h}}_{ii}}/{M}}^{2}}{\sum_{l \in \mathcal{L} \setminus \{i\}} {E_{\mathrm{d}}}/{\eta_{l}^{2}} \abs{{{\mathbf{h}}_{li}^{H} \hat{\mathbf{h}}_{ll}}/{M}}^{2} + 1 }.
\end{split}
\end{equation}
Substituting \eqref{equ_chan_los_nlos} into \eqref{equ_SINR_DL_UAV} and using \eqref{equ_h_asymptotic}, the asymptotic downlink receive SINR for the UAV can be similarly derived as
\begin{equation}\label{equ_SINR_DL_UAV_M}
\mathrm{SINR}^{\mathrm{d}}_{i} \xrightarrow[M \to \infty]{} \frac{E_{\mathrm{d}} \beta^{2}_{ii}/\eta_{i,\infty}^{2}}{\sum_{l \in \mathcal{L} \setminus \{i\}} E_{\mathrm{d}} \beta_{li}^{2} /\eta_{l,\infty}^{2} + 1}.
\end{equation}
On the other hand, for a GUE, its received signal is expressed as
\begin{equation}
x_{j}^{\mathrm{d}} = \sqrt{{p_{\mathrm{d}}}} \mathbf{h}_{jj}^{T} \underline{{\mathbf{w}}}_{j} x_{j}  + {n}^{\mathrm{d}}_{j},
\end{equation}
where $ {n}^{\mathrm{d}}_{j} $ denotes the receiver noise with $ {n}^{\mathrm{d}}_{j} \sim \mathcal{CN}(0,1) $.
Then the downlink receive SINR of each GUE is given by
\begin{equation}\label{equ_SINR_DL_GUE}
\mathrm{SINR}^{\mathrm{d}}_{j} = {\frac{E_{\mathrm{d}}}{\eta_{j}^{2}} \abs{\frac{\mathbf{h}_{jj}^{H} \hat{\mathbf{h}}_{jj}}{M}}^{2}} \xrightarrow[M \to \infty]{} {E_{\mathrm{d}}} \beta^{2}_{jj}/\eta_{j,\infty}^{2}.
\end{equation}
Due to the PC in the uplink channel estimation, each BS in the downlink data transmission fails to steer its beam directly towards its associated user. Consequently, as suggested by \eqref{equ_SINR_DL_UAV_M} and \eqref{equ_SINR_DL_GUE}, each user (regardless of GUE or UAV) suffers a certain signal power loss, while each UAV suffers additional interference from all the other $ K - 1 $ BSs.
%
\subsection{Performance Comparison Before versus After PDC}\label{sec_sys_cmp}
Fortunately, if each BS can detect the LoS interference from all non-associated UAVs (i.e., the second term in (\ref{equ_chan_los_nlos})), the interference that users suffer in both the uplink and downlink can be eliminated completely. 
Denote by $ \mathcal{A}_{l} = \{ \mathbf{a}(\theta_{lk},\phi_{lk})| k \in {\mathcal{U}_{l}} \} $ the set consisting of the steering vectors associated with all interfering UAVs, where $ (\theta_{lk},\phi_{lk}) $'s are their AoAs.
Then we can obtain the channel estimate after (perfect) PDC as $ \hat{\mathbf{h}}_{ll}^{\mathrm{D}} \in \mathbb{C}^{M \times 1} $ by removing the UAVs' interference term in (\ref{equ_chan_los_nlos}), i.e.,
\begin{equation}\label{equ_chn_prj}
\hat{\mathbf{h}}_{ll}^{\mathrm{D}} = \mathbf{P}^{\perp}_{l} \hat{\mathbf{h}}_{ll} = \mathbf{P}^{\perp}_{l} \mathbf{{h}}_{ll} +  \mathbf{P}^{\perp}_{l} \mathbf{n}_{l},
\end{equation}
where $ \mathbf{P}^{\perp}_{l} = \mathbf{I} - \mathbf{A}_{l} \mathbf{A}_{l}^{\dagger} $ and $ \mathbf{A}_{l} \in \mathbb{C}^{M \times |{\mathcal{U}_{l}}|} $ is the matrix with columns being all the elements in $ \mathcal{A}_{l} $. 
$(\cdot)^{\dagger} $ denotes the pseudo inverse, and $ |{\mathcal{U}_{l}}| $ denotes the cardinality of the set $ {\mathcal{U}_{l}} $. Define $ \overline{\mathbf{w}}_{l} = {\hat{\mathbf{h}}_{ll}^{\mathrm{D}}}/({\hat{\eta_{l}}\sqrt{M}}) $ and $ \underline{{\mathbf{w}}}_{l} = {(\hat{\mathbf{h}}_{ll}^{\mathrm{D}})^{*}}/({\hat{\eta_{l}}\sqrt{M}}) $ as the new (receive) combining and (transmit) precoding vectors after PDC, respectively, with $ \hat{\eta_{l}} = {\|{\hat{\mathbf{h}}_{ll}^{\mathrm{D}}}\|}/{\sqrt{M}} $. Then the asymptotic SINRs of the uplink and downlink data transmissions after PDC are given by
\begin{equation}\label{equ_SINR_DC}
	\scalebox{0.95}{$\hat{\mathrm{SINR}}^{\mathrm{u}}_{l} 
	 \xrightarrow[M \to \infty]{} E_{\mathrm{u}} \beta^{2}_{ll}/\hat{\eta}_{l,\infty}^{2},~
	\hat{\mathrm{SINR}}^{\mathrm{d}}_{l}  \xrightarrow[M \to \infty]{} E_{\mathrm{d}} \beta^{2}_{ll}/\hat{\eta}_{l,\infty}^{2}$},
\end{equation}
where $ {\hat{\eta}}_{l,\infty}^{2} = \lim\limits_{M \to \infty} {\hat{\eta}}_{l}^{2} = \beta_{ll} + {1}/({\tau p_{\mathrm{p}}}), \forall l $.
It is observed from (\ref{equ_SINR_DC}) that after PDC, the UAV-caused interference in the uplink is eliminated for all users, and in the downlink, all users will be free of any signal power loss as well as interference.
\begin{table} 
	\centering 
	\caption{Asymptotic SINR in Data Transmission} 
	\renewcommand{\arraystretch}{1}
	\scalebox{1}{
		\begin{tabular}{c|c|c|c|c}
			\hline
			& \multicolumn{2}{c|}{UL} & \multicolumn{2}{c}{DL} \\		
			\cline{2-5}
			& UAV & GUE & UAV & GUE \\
			\hline
			Before PDC & $ \frac{1}{K_{\mathrm{u}}-1} $ & $ \frac{1}{K_{\mathrm{u}}} $ & $\frac{1}{K-1} $ & $ \frac{E_{\mathrm{d}}\beta_{jj}}{K_{\mathrm{u}}+1} $\\
			\hline
			After PDC & $ E_{\mathrm{u}} \beta_{ll} $ & $ E_{\mathrm{u}} \beta_{ll} $ & $ E_{\mathrm{d}} \beta_{ii} $ & $ E_{\mathrm{d}} \beta_{jj} $\\			
			\hline
		\end{tabular}%
	}
	\label{tab_sinr}
\end{table}

Next, to draw further insights for the high SNR regime, we assume that $ E_{\mathrm{u}} \beta_{ll} \gg 1 $, $ E_{\mathrm{d}} \beta_{ll} \gg 1, \forall l $, and $ p_{\mathrm{p}} \gg 1 $. Furthermore, it is assumed that for any BS $ l $, $ \beta_{lk} \simeq \beta_{ll}, \forall k \in \mathcal{U}_{l} $, in the uplink and $ \beta_{lk} \simeq \beta_{0}, \forall k \in \mathcal{U}_{l} \cup \{l\}, \forall l \in \mathcal{L} $, in the downlink, by ignoring the distance differences from different UAVs to BS $ l $. Then, from \eqref{equ_SINR_UL_GUE_M}, \eqref{equ_SINR_DL_UAV_M}, \eqref{equ_SINR_DL_GUE} and (\ref{equ_SINR_DC}), we can obtain the asymptotic SINRs for the uplink and downlink, respectively, shown in Table \ref{tab_sinr}. The main insights are highlighted as follows.
\begin{itemize}
	\item In the uplink, the receive SINRs of users before PDC are bounded by the number of interfering UAVs (i.e., $ K_{\mathrm{u}} $ for  GUE and $ K_{\mathrm{u}} - 1 $ for UAV), regardless of $ E_{\mathrm{u}} $. While after PDC, their SINRs can increase with $ E_{\mathrm{u}} $ due to the interference elimination. 
	\item In the downlink, the receive SINR of UAV users before PDC is bounded by the number of users sharing the given pilot (i.e., $ K -1 $), which is also regardless of $ E_{\mathrm{d}} $. In contrast, the receive SINR of GUE users increases with $ E_{\mathrm{d}} $ thanks to the negligible interference from far-apart non-associated BSs, but it decreases proportionally to $ {1}/({K_{\mathrm{u}}+1}) $ due to power loss. After PDC, power loss and interference are both eliminated, thus rendering the SINRs of all users to increase proportionally with $ E_{\mathrm{d}} $.
\end{itemize}

Motivated by the above results on the significant SINR performance gains after versus before PDC, we propose a practical scheme to resolve the UAV-induced PC next. 
%
\section{Proposed Algorithm}\label{sct_alg_t}
The key for mitigating the UAV-induced PC is to detect the interfering UAV LoS signals and extract their AoAs.
To this end, we first propose a successive LoS component detector by exploiting the AoA-dependent characteristics of LoS links.
Then, for each GUE user, the LoS interference can be identified thanks to the elevation angle separation between GUEs and UAVs. Whereas for each UAV user, its AoA cannot be separated from those of the other UAVs (if $ K_{\mathrm{u}} \geq 2 $). Thus, we propose to let each UAV transmit a different pilot in the next training block to help identify its AoA if the PC is detected during the first training block.
%
\subsection{PDC for GUE User}\label{sct_alg_gue}
First, we tackle the PC problem for GUEs. We assume that each BS knows that its associated user is a UAV or GUE prior to the uplink channel estimation (which can be realized in the preceding user-BS association stage). 
As assumed in Section \ref{sct_sysmdl}, we consider UAVs that fly higher than the BSs, thus we only need to consider the angle range in which the AoAs of the interfering UAV LoS signals possibly reside. The possible AoA ranges in the elevation and azimuth dimensions are respectively $ \Theta = [0,\pi/2]$ and $ \Phi = [-\pi, \pi] $.
In addition, since GUEs are located in practice lower than BSs, the LoS paths (if any) of their channels with the BSs will be out of the above ranges. Thus, for each GUE user, all the LoS components detected in the above ranges are considered as interference.

To detect all LoS components in the above range, each BS can perform a successive detection procedure by iteratively detecting and removing the strongest LoS component in its estimated channel until no additional strong LoS component can be found. 
Specifically, we first discretize the search ranges $  \Theta $ and $ \Phi $ as $ \theta_{\bar{m}} = {\bar{m} \pi}/({2 N_{\theta}}) , \bar{m} = 0,\cdots, N_{\theta}-1$, and $ \phi_{\bar{n}} = {\bar{n}\pi}/{N_{\phi}} - \pi, \bar{n} = 0,\cdots, N_{\phi}-1$, with $ N_{\theta} $ and $ N_{\phi} $ denoting the number of grids at $ \theta $ and $ \phi $ directions, respectively.\footnote{Generally, $ N_{\theta} $ and $ N_{\phi} $ strike a balance between quantization error and computational complexity. In practice, $ N_{\theta} $ and $ N_{\phi} $ should be set satisfying $ {\pi}/({2 N_{\theta}}) < \theta_{3\mathrm{dB}} $ and $ {\pi}/{N_{\phi}} < \phi_{3\mathrm{dB}} $, with $ \theta_{3\mathrm{dB}} $ and $ \phi_{3\mathrm{dB}} $ denoting 3-dB beam-width of the BS UCA.} 
Next, the BS associated with GUE $ j $ performs matched filtering over the (effective) channel estimate at each quantized direction $ (\theta_{\bar{m}},\phi_{\bar{n}}) $, and denote by $ T_{\bar{m},\bar{n}}^{q} $ the corresponding output in the $ q $-th round of detection, $ q \geq 1 $. Then, the decision on whether a LoS component is present or not in the $ q $-th round of detection is made according to
\begin{equation}\label{equ_mf_def}
\setstackgap{L}{.5\baselineskip}
\max_{\theta_{\bar{m}} \in \Theta,\phi_{\bar{n}} \in \Phi} \left \{ T_{\bar{m},\bar{n}}^{q} = \frac{1}{{M}} \left | \mathbf{a}^{H}(\theta_{\bar{m}},\phi_{\bar{n}}) \hat{\mathbf{h}}^{q}_{jj} \right |^{2} \right \}
\underset{\mathcal{H}_{0}}{\overset{\mathcal{H}_{1}}{\gtrless}} \zeta^{q},
\end{equation}
where $ \zeta^{q} $ is the threshold set as $ \zeta^{q} = \kappa \sum_{\bar{m},\bar{n}} T_{\bar{m},\bar{n}}^{q} /(N_{\theta} N_{\phi}),
 $
%
%
with $ \kappa $ being a positive constant. We then discuss the following two cases in \eqref{equ_mf_def}:
\begin{itemize}
	\item If $ \mathcal{H}_{1} $ holds, a LoS component is declared to be present, with its AoA corresponding to the maximum value of $ T_{\bar{m},\bar{n}}^{q} $'s, denoted by $ (\theta^{q}_{jj}, \phi^{q}_{jj}) $. Note that this LoS component can be reconstructed as $ {\overline{{\mathbf{h}}}}^{q}_{jj} = \mu_{jj}^{q} \mathbf{a}(\theta^{q}_{jj}, \phi^{q}_{jj}) $, where $ \mu_{jj}^{q} = \mathbf{a}^{\dagger}(\theta^{q}_{jj}, \phi^{q}_{jj}) \hat{\mathbf{h}}^{q}_{jj} $ accounts for its path loss and phase rotation, given by the optimal solution of $ \min_{\mu} \| \hat{\mathbf{h}}^{q}_{jj} - \mu \mathbf{a}(\theta^{q}_{jj}, \phi^{q}_{jj}) \|^{2} $. Then $ \overline{{{\mathbf{h}}}}^{q}_{jj} $ can be removed from the current effective channel estimate, i.e.,
	\begin{equation}\label{equ_hath_prj}
	\hat{\mathbf{h}}^{q+1}_{jj} = \hat{\mathbf{h}}^{q}_{jj} - \overline{{{\mathbf{h}}}}^{q}_{jj}. 
	\end{equation}
	The updated effective channel estimate $ \hat{\mathbf{h}}^{q+1}_{jj} $ is then substituted into \eqref{equ_mf_def}-\eqref{equ_hath_prj} for the next round of detection. 
	\item If $ \mathcal{H}_{0} $ holds, $ \hat{\mathbf{h}}^{q}_{jj} $ is assumed to contain no more LoS component. Thus, we terminate the successive detection and 
	take $ \hat{\mathbf{h}}^{q}_{jj} $ as the final channel estimate, denoted by $ \hat{\mathbf{h}}^{\mathrm{D}}_{jj} = \hat{\mathbf{h}}^{q}_{jj} $, and define $ L_{j} = q - 1 $ as the number of detected LoS components. 	
\end{itemize}
Note that at the beginning of the above successive detection, we set $ q = 1 $ and $ \hat{\mathbf{h}}^{1}_{jj} = \hat{\mathbf{h}}_{jj} $.
%
\subsection{PDC for UAV User} \label{sct_alg_uav}
For each UAV user $ i $, its associated BS first performs the successive detection proposed in the previous subsection to detect all $ L_{i} $ strong LoS components, with $ L_{i} \geq 0 $. Define $ \mathcal{D}_{i} $ as the set consisting of all detected strong LoS components, where $ \mathcal{D}_{i} = \{\overline{\mathbf{h}}_{ii}^{s}| s = 1,\cdots, L_{i}\} $ if $ L_{i} \geq 1 $, and $ \mathcal{D}_{i} = \O $ otherwise. If $ |\mathcal{D}_{i}| = 0 $, it implies that no significant LoS channel is detected for UAV user $ i $ (which occurs with a very low probability in practice); while if $ |\mathcal{D}_{i}| = 1 $, then a unique LoS channel is detected for the UAV user. In both cases, we can set the channel estimate for UAV user $ i $ as $ \hat{\mathbf{h}}_{ii}^{\mathrm{D}} = \hat{\mathbf{h}}_{ii}^{\mathrm{1}} $.
However, if $ |\mathcal{D}_{i}| \geq 2 $, then PC is considered to have occurred, which needs to be resolved by further processing.
The key is to identify which LoS channel in $ \mathcal{D}_{i} $ is due to UAV user $ i $, which is challenging since there is no prior knowledge on the UAV users' locations assumed to be known at the BSs.

To solve this problem, we propose that the associated BS with UAV user $ i $ informs it to send a different pilot in the next training block. 
Define $ \underline{\mathcal{D}}_{i} $ as the set consisting of all $ \underline{L}_{i} \geq 0 $ strong LoS components detected in the second training block, where $ \underline{\mathcal{D}}_{i} = \{\overline{\underline{\mathbf{h}}}_{ii}^{s}| s = 1,\cdots,\underline{L}_{i} \} $ if $ \underline{L}_{i} \geq 1 $, and $ \underline{\mathcal{D}}_{i} = \O $ otherwise. Then the desired LoS channel of UAV user $ i $ can be identified with a high probability by comparing $ {\mathcal{D}}_{i} $ and $ \underline{\mathcal{D}}_{i} $.
In practice, the two groups of UAVs that share the same pilot with UAV user $ i $ in each of the two training blocks are different with a high probability, since we assume the BSs randomly assign pilots to their associated users independently and the set of pilots is practically large. However, both $ {\mathcal{D}}_{i} $ and $ \underline{\mathcal{D}}_{i} $ should contain the LoS channel of UAV user $ i $, which is very likely to be their only common element.
Denote by $ \Delta \mathcal{D}_{i} $ the set consisting of the LoS components that are approximately equal (i.e., the Euclidean distance of the two vectors is less than a given small constant) in $ \mathcal{D}_{i} $ and $ \underline{\mathcal{D}}_{i} $.
If $ |\Delta \mathcal{D}_{i}| = 1 $, the only element in $ \Delta \mathcal{D}_{i} $ is taken as the LoS channel of UAV user $ i $. 
In practice, this is the most likely case since it is generally of very low probability to have two UAVs that have similar AoAs as well as distances (channel gains) with the associated BS of UAV user $ i $, and are also assigned with identical (randomly selected) pilots during the two training blocks. 
Nevertheless, if the above low-probability event occurs which results in $ |\Delta \mathcal{D}_{i}| > 1 $, we can only assure that the LoS components in $ \mathcal{D}_{i}/\Delta \mathcal{D}_{i} $ are interferences. In both cases, we remove the LoS interferences in $ \mathcal{D}_{i}/\Delta \mathcal{D}_{i} $ from $ \hat{\mathbf{h}}_{ii}^{\mathrm{1}} $ to obtain the decontaminated channel estimate $ \hat{\mathbf{h}}_{ii}^{\mathrm{D}} $.
If $ |\Delta \mathcal{D}_{i}| = 0 $, then we fail to identify any interfering LoS channel for UAV user $ i $ and simply set $ \hat{\mathbf{h}}_{ii}^{\mathrm{D}} = \hat{\mathbf{h}}_{ii}^{\mathrm{1}} $ (albeit this is also very unlikely in practice).
\section{Numerical Results}\label{sct_sml}
This section provides numerical results to compare the performance before (bf.) and after (af.) applying the proposed PDC scheme. The ideal case where each BS knows the channels of all users is also shown as the performance upper bound. In addition, the PDC scheme proposed in \cite{SF_GongZJ_TWC_19} is included as a benchmark. The cellular network topology is shown in Fig. \ref{fig_sysmdl}. The given pilot is reused by $ K = 9 $ cells, each serving one user (UAV or GUE) for the given pilot. We set the pilot reuse factor $ R=7 $, due to which the GUE-induced PC is negligible for the considered setup. 
Each BS's height is 25 meters (m) and the cell radius is 500 m. The UCA is employed at each BS with $ M = 128 $.
The heights of UAVs are uniformly distributed between 25 m and 300 m, while the heights of GUEs are fixed to be 1.5 m. The transmit powers of users and BSs are 23 dBm and 46 dBm, respectively. The noise power spectrum density at the receiver is $ - $164 dBm/Hz including a 10 dB noise figure, and the system bandwidth is 10 MHz. In addition, the parameter related to the threshold in the successive detection is $ \kappa = 3 $.

\begin{figure}[t]
	\centering
	\includegraphics[draft=false, width=.9\linewidth]{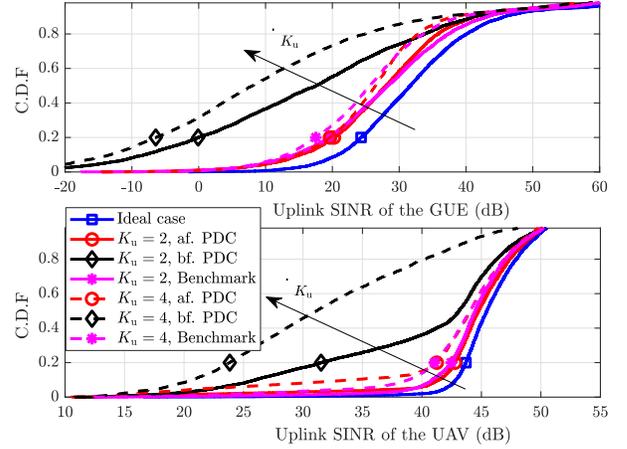}
	\vspace{-.3cm}
	\caption{The CDF for the SINR in the uplink. \label{fig_pfm_ul}}
\end{figure}

Fig. \ref{fig_pfm_ul} plots the empirical cumulative distribution function (CDF) for the SINR in the uplink data transmission. It is observed that the proposed PDC scheme achieves significant performance gains for both UAVs and GUEs and also obtains almost the same gains as the benchmark. Furthermore, in accordance with our analysis in Section \ref{sct_sysmdl}, when more UAVs are involved (i.e., $ K_{\mathrm{u}} $ is larger), more severe SINR degradation is resulted for both UAVs and GUEs before applying the PDC scheme.
 
\begin{figure}[t]
	\centering
	\includegraphics[draft=false, width=.9\linewidth]{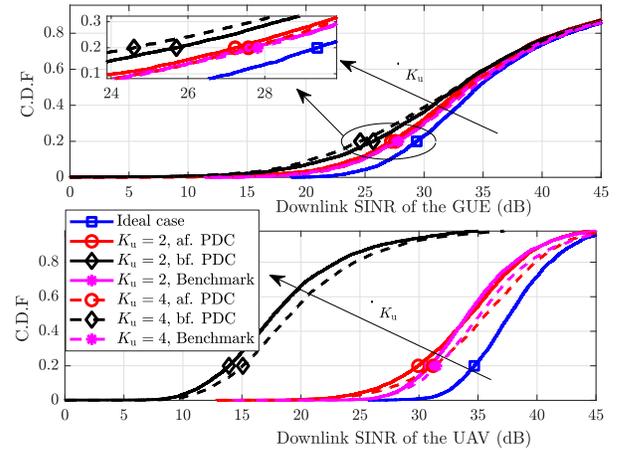}
	\vspace{-.3cm}
	\caption{The CDF for the SINR in the downlink. \label{fig_pfm_dl}}
\end{figure}
	
Fig. \ref{fig_pfm_dl} plots the CDF for the SINR in the downlink data transmission. One can observe that the UAVs suffer from more severe SINR degradation than GUEs before applying the proposed PDC scheme, and significant performance gains are achieved after applying the PDC scheme, especially for UAV users. These results are also consistent with our analysis in Section \ref{sct_sysmdl}. 
\section{Conclusion}
This letter addresses a new and challenging PC issue in massive MIMO networks communicating with UAVs. We first derive the SINRs of UAVs and GUEs before and after the PDC for both the uplink and downlink, and unveil their large performance gaps. 
Then we propose practical algorithms to resolve the UAV-induced PC for both GUEs and UAVs by exploiting their different channel characteristics with the BSs.
Numerical results show significant SINR performance improvement in both uplink/downlink data transmission after applying the proposed PDC algorithms. 
\vspace{-2.5mm}
\bibliographystyle{IEEEtran}
\bibliography{mybib}

\end{document}